\documentclass[11pt,letterpaper]{article}

\usepackage{fullpage,amsmath,amsthm}
\usepackage{mathtools,authblk}
\usepackage[colorinlistoftodos]{todonotes}
\usepackage{algorithm, algorithmic}

\usepackage{graphicx,subfig}
\usepackage{threeparttable,tabularx,caption}
\newcolumntype{b}{X}
\newcolumntype{s}{>{\hsize=.68\hsize}X}
\newcolumntype{m}{>{\hsize=.8\hsize}X}

\newtheorem{definition}{Definition}
\newtheorem{theorem}{Theorem}
\newtheorem{corollary}{Corollary}
\newtheorem{lemma}{Lemma}
\newcommand{\infdiv}[3]{#1\,|#2|\,#3}

\newcommand{\floor}[1]{\lfloor #1 \rfloor}
\newcommand{\ceil}[1]{\lceil #1 \rceil}

\title{Select and Permute: An Improved Online Framework for Scheduling to Minimize Weighted Completion Time
\footnote{All authors performed this work at the University of Maryland, College Park, under the support of NSF REU Grant CNS 156019. We would also like to thank An Zhu and Google for their support, and the LILAC program at Bryn Mawr College.}
}

\author[1]{Samir Khuller\thanks{\texttt{samir@cs.umd.edu}}}
\author[2]{Jingling Li\thanks{\texttt{jinglingli1024@gmail.com}}}
\author[3]{Pascal Sturmfels\thanks{\texttt{psturm@umich.edu}}}
\author[4]{Kevin Sun\thanks{\texttt{kevin.sun@rutgers.edu}}}
\author[1]{Prayaag Venkat\thanks{\texttt{pkvasv@gmail.com}}}
\affil[1]{University of Maryland, College Park, MD 20742}
\affil[2]{Bryn Mawr College, Bryn Mawr, PA 19010}
\affil[3]{University of Michigan, Ann Arbor, MI 48109}
\affil[4]{Rutgers University, New Brunswick, NJ 08901}

\begin{document}
\maketitle

\begin{abstract}
In this paper, we introduce a new online scheduling framework for minimizing total weighted completion time in a general setting. The framework is inspired by the work of Hall et al.~\cite{hall1997scheduling} and Garg et al.~\cite{garg2007order}, who show how to convert an offline approximation to an online scheme. Our framework uses two offline approximation algorithms---one for the simpler problem of scheduling without release times, and another for the \emph{minimum unscheduled weight problem}---to create an online algorithm with provably good competitive ratios.

We illustrate multiple applications of this method that yield improved competitive ratios. Our framework gives algorithms with the best or only-known competitive ratios for the concurrent open shop, coflow, and concurrent cluster models. We also introduce a randomized variant of our framework based on the ideas of Chakrabarti et al.~\cite{chakrabarti1996improved} and use it to achieve improved competitive ratios for these same problems.

\end{abstract}

%


\section{Introduction}
\label{sec:intro}
Modern computing frameworks such as MapReduce, Spark, and Dataflow have emerged as essential tools for big data processing and cloud computing. To exploit large-scale parallelism, these frameworks act in several computation stages, which are interspersed with intermediate data transfer stages. During data transfer, results from computations must be efficiently scheduled for transfer across clusters so that the next computation stage can begin.


The coflow model~\cite{chowdhury2012coflow,chowdhury2014efficient} and the concurrent cluster model~\cite{hung2015scheduling,murray2016scheduling} were introduced to capture the distributed processing requirements of jobs across many machines. In these models, the objective of primary theoretical and practical interest is to minimize average job completion time~\cite{ahmadi2016coflows,chowdhury2012coflow,chowdhury2014efficient,khuller2016brief,murray2016scheduling,qiu2015minimizing}. The \emph{concurrent open shop problem}, a special case of the above models, has emerged as a key subroutine for designing better approximation algorithms~\cite{ahmadi2016coflows,khuller2016brief,murray2016scheduling}. 

There has been a lot of work studying offline algorithms for these problems (see \cite{ahmadi2016coflows,khuller2016brief,qiu2015minimizing} for the coflow model,~\cite{murray2016scheduling} for the concurrent cluster model, and~\cite{chen2000supply,garg2007order,mastrolilli2010minimizing,wang2007customer} for the concurrent open shop model), but in real-world applications, jobs often arrive in an online fashion, so studying online algorithms is critical for accurate modeling of data centers.

Hall et al.~\cite{hall1997scheduling} proposed a general framework which converts offline scheduling algorithms to online ones. Inspired by this result, we introduce a new online framework that improves upon the online algorithms of Garg et al.~\cite{garg2007order} for concurrent open shop and also gives the first algorithms with constant competitive ratios for other multiple-machine scheduling settings. 

\subsection{Formal Problem Statement}
In the concurrent open shop setting, the problem is to schedule a set of jobs with machine-dependent components on a set of machines. Let $J = \{1, \ldots, n\}$ denote the set of jobs and $M = \{1, \ldots, m\}$ denote the set of machines. Each job $j$ has one component for each of the $m$ machines. For each job $j$, we denote the processing time of the component on machine $i$ as $p_{ij}$, its release time as $r_j$, and its weight as $w_j$. The different components of each job can be processed concurrently and in any order, as long as no component of job $j$ is processed before $r_j$. Job $j$ is complete when all of its components have been processed; we denote its completion time by $C_j$. Our goal is to specify a schedule of the jobs on the machines that minimizes $\sum_{j \in J} w_j C_j$; see Fig.~\ref{fig:schedules} for an example.

\begin{figure}
\centering
\subfloat[An instance of concurrent open shop.]{\includegraphics[width=0.4\textwidth]{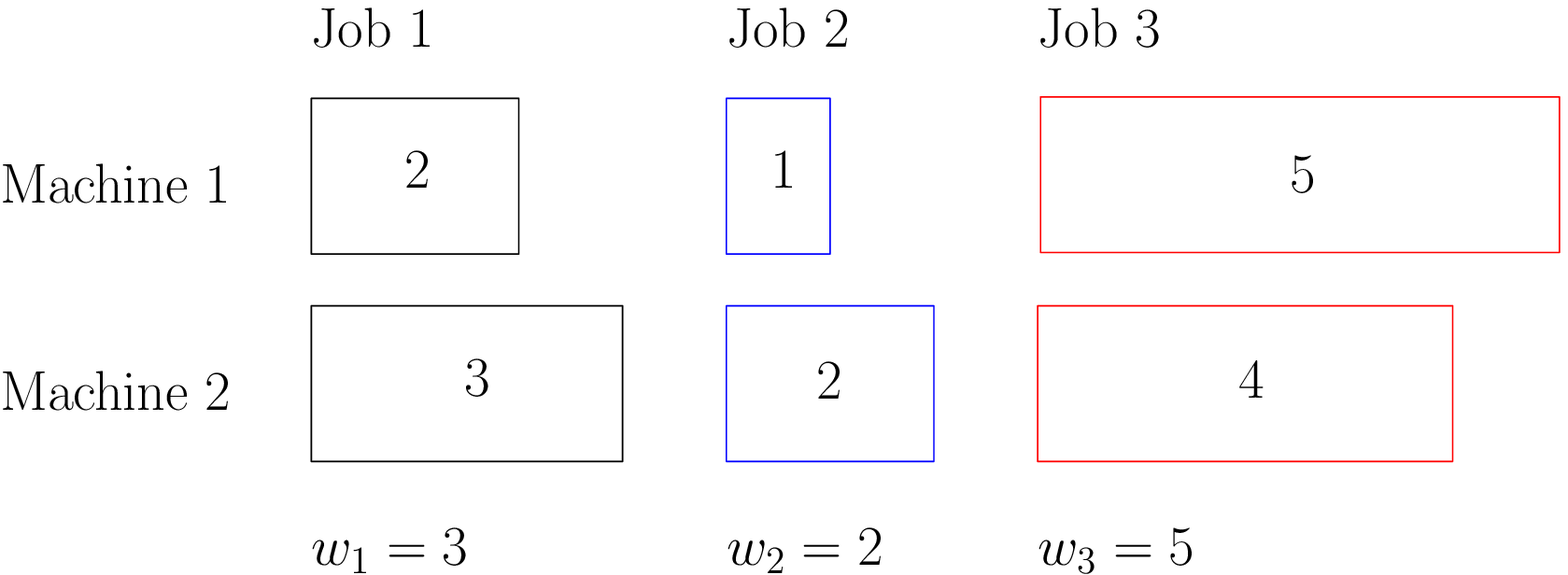}\label{fig:subfig_schedule}}
\hspace{10pt}
\subfloat[The optimal schedule]{\includegraphics[width=0.45\textwidth]{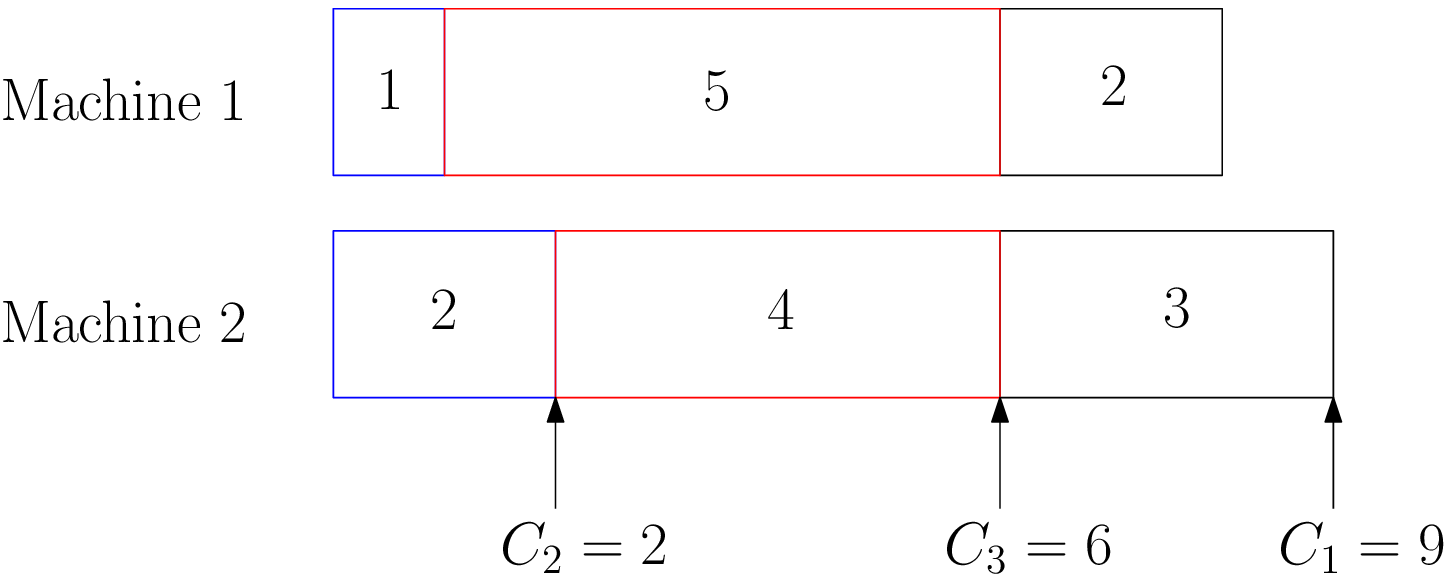}\label{fig:subfig_opt_schedule}}
\caption{All jobs are released at the same time, and the processing requirement for each job-machine combination is specified inside the blocks.}
\label{fig:schedules}
\end{figure}

We follow the 3-field $\infdiv{\alpha}{\beta}{\gamma}$ notation (see \cite{graham1979}) for scheduling problems, where $\alpha$ denotes the scheduling environment, $\beta$ denotes the job characteristics, and $\gamma$ denotes the objective function. As stated above, we focus on the case where $\gamma = \sum_j{w_jC_j}$. In accordance with the notation of~\cite{hall1997scheduling,mastrolilli2010minimizing,murray2016scheduling}, we let $\alpha = PD$ denote the concurrent open shop setting and $\alpha = CC$ denote the concurrent cluster setting, see below for definitions.

\subsection{Related Work}
The concurrent open shop model is a relaxation of the well-known open shop model that allows components of the same job to be processed in parallel on different machines. Roemer~\cite{Roemer2006} showed that $\infdiv{PD}{}{\sum_j w_j C_j}$ is NP-hard and after several successive approximation hardness results~\cite{bansal2010inapproximability,mastrolilli2010minimizing}, Sachdeva and Saket~\cite{sachdeva2013optimal} showed that it is not approximable within a factor less than $2$ unless P = NP, even when job release times are identical. For this model, Wang and Cheng~\cite{wang2007customer} gave a $\frac{16}{3}$-approximation algorithm. This was later improved to a 2-approximation for identical job release times \cite{chen2000supply,garg2007order,leung2007scheduling,mastrolilli2010minimizing}, matching the above lower bound, and a 3-approximation for arbitrary job release times \cite{ahmadi2016coflows,garg2007order,leung2007scheduling}.

In the online setting, Hall et al.~\cite{hall1997scheduling} introduced a general framework that improved the best-known approximation guarantees for several well-studied scheduling environments. They showed that the existence of an offline \emph{dual} $\rho$-approximation yields an online $4\rho$-approximation, where a dual $\rho$-approximation is an algorithm that packs as much weight of jobs into a time interval of length $\rho D$ as the optimal algorithm does into an interval of length $D$. Furthermore, they showed that when $m=1$, a local greedy ordering of jobs yields further improvements.
While the framework of Hall et al.~\cite{hall1997scheduling} is entirely deterministic, Chakrabarti et al.~\cite{chakrabarti1996improved} gave a randomized variant with an improved competitive ratio guarantee. Specifically, they showed a dual-$\rho$ approximation algorithm can be converted to an expected $2.89\rho$-competitive online scheduling algorithm in the same setting, improving upon the $4\rho$ competitive ratio of Hall et al.~\cite{hall1997scheduling}.


The online version of $\infdiv{PD}{}{\sum_j w_j C_j}$ was first studied by Garg et al.~\cite{garg2007order}. They noted that applying the framework of~\cite{hall1997scheduling} was not straightforward, so they focused on minimizing the weight of unscheduled jobs rather than maximizing the weight of scheduled jobs. Using a similar approach to that of Hall et al.~\cite{hall1997scheduling}, they gave an exponential-time 4-competitive algorithm and a polynomial-time 16-competitive algorithm for the online version of $\infdiv{PD}{}{\sum_j w_j C_j}$.

The coflow scheduling model was first introduced as a networking abstraction to model communications in datacenters~\cite{chowdhury2012coflow,chowdhury2014efficient}. In the coflow scheduling problem, the goal is to schedule a set of \emph{coflows} on a non-blocking switch with $m$ input ports and $m$ output ports, where any unused input-output ports can be connected via a path through unused nodes regardless of other existing paths. Each coflow is a collection of parallel flow demands that specify the amount of data that needs to be transferred from an input port to an output port. 

For $\infdiv{\mbox{Coflow}}{r_j=0}{\sum_j{w_jC_j}}$, Qiu et al.~\cite{qiu2015minimizing} gave deterministic $\frac{64}{3}$ and randomized $(8 + \frac{16 \sqrt{2}}{3})$ approximation algorithms. For arbitrary release times, they gave deterministic $\frac{67}{3}$ and randomized $9 + \frac{16 \sqrt{2}}{3}$ approximation algorithms. Khuller and Purohit~\cite{khuller2016brief} later improved these deterministic approximations to 8 and 12 for identical and arbitrary release times respectively, and also gave a randomized $(3 + 2\sqrt{2})$-approximation algorithm for identical release times. Recently, Ahmadi et al.~\cite{ahmadi2016coflows} gave a deterministic 4-approximation and 5-approximation for identical and arbitrary release times, respectively. To the best of our knowledge, there are no known constant-factor competitive algorithms for online coflow scheduling, although Li et al.~\cite{li2016efficient} have given a $O(n \ln m)$-competitive algorithm when all coflow weights are equal to 1.

Finally, we mention the \emph{concurrent cluster} model recently introduced by Murray et al.~\cite{murray2016scheduling}. The concurrent cluster model generalizes the concurrent open shop model by replacing each machine by a cluster of machines, where different machines in the same cluster may have different processing speeds. Each job still has $m$ processing requirements, but this requirement can be fulfilled by any machine in the corresponding cluster. Murray et al.~\cite{murray2016scheduling} give the first constant-factor approximations for minimizing total weighted completion time via a reduction to concurrent open shop and a list-scheduling subroutine.

\subsection{Paper Outline and Results}
In Sect.~\ref{sec:framework}, we introduce a general framework for designing online scheduling algorithms for minimizing total weighted completion time. The framework divides time into intervals of geometrically-increasing size, and greedily ``packs'' jobs into each interval, and then imposes a locally-determined ordering of the jobs within each interval. It is inspired by the framework of Hall et al.~\cite{hall1997scheduling} and an adaptation by Garg et al.~\cite{garg2007order}.

In Sect.~\ref{sec:cos}, we apply our framework to $\infdiv{PD}{}{\sum_j w_j C_j}$. We show that an offline exponential-time algorithm that optimally solves $\infdiv{PD}{r_j = 0}{\sum_j w_j C_j}$ yields an exponential-time 3-competitive algorithm for $\infdiv{PD}{}{\sum_j w_j C_j}$. We also combine the algorithms given by Garg et al.~\cite{garg2007order} and Mastrolilli et al.~\cite{mastrolilli2010minimizing} to create a polynomial-time 10-competitive algorithm for $\infdiv{PD}{}{\sum_j w_j C_j}$. We conclude Sect.~\ref{sec:cos} by giving a polynomial-time $(3+\epsilon)$-competitive algorithm when the number of machines $m$ is fixed.

In Sect.~\ref{sec:coflow}, extending the ideas of Sect.~\ref{sec:cos}, we apply our framework to online coflow scheduling to design an exponential-time 6-competitive algorithm, and a polynomial-time 12-competitive algorithm. 


Section~\ref{sec:rand_framework}, describes an extension of techniques of Chakrabarti et al.~\cite{chakrabarti1996improved} that produces a randomized variant of our framework that yields better competitive ratio guarantees than the deterministic version. Section~\ref{sec:ccs} describes the concurrent cluster model of Murray et al.~\cite{murray2016scheduling}; we show that extending subroutines used for the concurrent open shop setting yields an online 19-competitive algorithm via our framework. Sections~\ref{sec:cos_exp_details} and~\ref{sec:cos_4_poly_fixed_m} describe the subroutines used in Sect.~\ref{sec:cos}.

\begin{center}
\begin{threeparttable}
\captionsetup{width=0.95\textwidth}
\caption{A summary of online approximation guarantees and the best-known previous results, where $m$ denotes the number of machines, $\epsilon$ is arbitrarily small, and ``-'' indicates the absence of a relevant result. The two numbers in each entry of the ``Our ratios'' column denote the competitive and expected ratio of our deterministic and randomized algorithms, respectively.}
\label{table:results}
\begin{tabularx}{0.98\textwidth}{b|b|m|s}
\hline
Problem & Running time & Our ratios & Previous ratio \\
\hline
$\infdiv{PD}{}{\sum_j{w_jC_j}}$ & polynomial & 10, 7.78 & 16~\cite{garg2007order} \\
\hline
$\infdiv{PD}{}{\sum_j{w_jC_j}}$ & exponential & 3, 2.45 & 4~\cite{garg2007order} \\
\hline
$\infdiv{PD}{}{\sum_j{w_jC_j}}$ & polynomial, fixed $m$ & $3+\epsilon$, $2.45+\epsilon$ & - \\
\hline
$\infdiv{\mbox{Coflow}}{}{\sum_j{w_jC_j}}$ & polynomial & 12, 9.78 & - \\
\hline
$\infdiv{\mbox{Coflow}}{}{\sum_j{w_jC_j}}$ & exponential & 6, 5.45 & - \\
\hline
$\infdiv{CC}{}{\sum_j{w_jC_j}}$ & polynomial & 19, 14.55 & - \\
\hline
\end{tabularx}
\end{threeparttable}
\end{center}
\section{A Minimization Framework for Online Scheduling}
\label{sec:framework}
In this section, we introduce our framework for online scheduling problems. To motivate the key ideas of this section, we begin by briefly reviewing the work of Hall et al.~\cite{hall1997scheduling} and Garg et al.~\cite{garg2007order}.

\subsection{The maximization framework of Hall et al.~\cite{hall1997scheduling}}
\label{sec:hall_framework}
The framework of Hall et al.~\cite{hall1997scheduling} divides the online problem into a sequence of offline \emph{maximum scheduled weight} problems, each of which is solved using an offline \emph{dual} approximation algorithm.

\begin{definition}[Maximum scheduled weight problem (MSWP) \cite{hall1997scheduling}] Given a set of jobs, a non-negative weight for each job, and a deadline $D$, construct a schedule that maximizes the total weight of jobs completed by time $D$.
\end{definition}


\begin{definition}[Dual $\rho$-approximation algorithm \cite{hall1997scheduling}]
An algorithm for the MSWP is a \emph{dual $\rho$-approximation algorithm} if it constructs a schedule of length at most $\rho D$ and has total weight at least that of the schedule which maximizes the weight of jobs completed by $D$.
\end{definition}

Fix a scheduling environment and suppose we have a dual $\rho$-approximation for the MSWP. We divide time into intervals of geometrically-increasing size by letting $t_0 = 0$ and $t_k = 2^{k-1}$ for $k = 1, \ldots, L$ where $L$ is large enough to cover the entire time horizon. At each time $t_k$, let $R(t_k)$ denote the set of jobs that have arrived by $t_k$ but have not yet been scheduled. We run the dual $\rho$-approximation algorithm on $R(t_k)$ with deadline $D = t_{k+1} - t_k = t_k$. In the output schedule, we take only jobs which complete by $\rho D$ and schedule them in the interval starting at time $\rho t_k$. Hall et al.~\cite{hall1997scheduling} show that this framework produces an online $4\rho$-competitive algorithm.

\subsection{The minimum unscheduled weight problem of Garg et al.~\cite{garg2007order}}

Garg et al.~\cite{garg2007order} sought to apply the framework of Hall et al.~\cite{hall1997scheduling} to the concurrent open shop setting. They noted that devising a dual-$\rho$ approximation algorithm for concurrent open shop was difficult, so they instead proposed a variant of the MSWP. The definitions below generalize those used by Garg et al.~\cite{garg2007order} to arbitrary scheduling problems.

\begin{definition}[Minimum unscheduled weight problem (MUWP)]
Given a set of jobs, a non-negative weight for each job, and a deadline $D$, find a subset of jobs $S$ which can be completed by time $D$ and minimizes the total weight of jobs not in $S$. We call this quantity the \emph{unscheduled weight}.
\end{definition} 


\begin{definition}[$(\alpha,\beta)$-approximation algorithm]
An algorithm for the MUWP is an \emph{$(\alpha,\beta)$-approximation} if it finds a subset of jobs which can be completed by $\alpha D$ and has unscheduled weight at most $\beta$ times that of the subset of jobs with minimum unscheduled weight that completes by $D$.
\end{definition}

Note that a dual $\rho$-approximation for the MSWP is a $(\rho,1)$-approximation for the MUWP. With these definitions, Garg et al.~\cite{garg2007order} established constant-factor approximations for $\infdiv{PD}{}{\sum_j{w_jC_j}}$.

\subsection{A minimization framework}
\label{sec:min_framework}
We now describe a new framework inspired by the ideas of Hall et al.~\cite{hall1997scheduling} and Garg et al.~\cite{garg2007order}. For the settings we consider, previous online algorithms do not impose any particular ordering of jobs within each interval, which can lead to schedules with poor local performance. In our framework, we combine an $(\alpha,\beta)$-approximation algorithm for MUWP and a $\gamma$-approximation to the offline version of the scheduling problem with identical release times to address this issue. Our framework generalizes and merges some of the ideas of Garg et al.~\cite{garg2007order} and Hall et al.~\cite{hall1997scheduling} so that it can be applied to a broader class of scheduling problems and only requires blackbox access to approximation algorithms for simpler problems.

As in the works of Hall et al.~\cite{hall1997scheduling} and Garg et al.~\cite{garg2007order}, we assume that all processing times are at least 1. This is to avoid the extreme scenario that a single job of size $\epsilon \ll 1$ arrives just after time 0, and our framework waits until time 1 to schedule, thus leading to arbitrarily large competitive ratio.

Let $W$ denote the total weight of all the jobs in $J$, and let $W_{\tau}^{\mathcal{A}}$ ($W_{\tau}^{\mathcal{OPT}}$) denote the total weight of jobs that complete after time $\tau$ by our algorithm $\mathcal{A}$ (by the optimal algorithm $\mathcal{OPT}$). Note that $W_{\tau}^{\mathcal{A}}, W_{\tau}^{\mathcal{OPT}}$ include the weight of jobs not yet released at time $\tau$. Let $\tau_0 = 0$, and for $k \geq 1$, let $\tau_{k} = 2^{k-1}$, $I_k$ denote the $k^\text{th}$ interval $[\tau_k, \tau_{k+1})$, $\alpha I_k$ denote $[\alpha \tau_k, \alpha \tau_{k+1})$, and $R(\tau_k)$ denote the set of jobs released but not yet scheduled before $\tau_k$ by $\mathcal{A}$. 

Our online algorithm $\mathcal{A}$ works as follows. At each $\tau_k$, run an $(\alpha,\beta)$-approximation algorithm on $R(\tau_k)$ with deadline $D = \tau_{k+1} - \tau_k$. Schedule the output set of jobs in $\alpha I_k$ using the offline $\gamma$-approximation algorithm.\footnote{We make the critical assumption that the offline $\gamma$-approximation algorithm does not increase the makespan of the given subset of jobs, so as to ensure that the schedule fits inside of $\alpha I_k$. For the scheduling models studied in this paper, this assumption will indeed hold. In fact, if it can be shown that the $\gamma$-approximation algorithm also approximates the makespan criteria within some factor $\mu$, then it is straightforward to incorporate this into the model, at the expense of an additional $\mu$ factor in the approximation guarantee. For example, Chakrabarti et al.~\cite{chakrabarti1996improved} provide bicriteria approximation algorithms for the total weighted completion time and makespan objective functions.}

\begin{theorem} \label{thm:alphabetagamma}
Algorithm $\mathcal{A}$ is $(2\alpha \beta + \gamma)$-competitive, with an additive $\alpha W$ term.
\end{theorem}

To prove Theorem~\ref{thm:alphabetagamma}, we first show that at each time step, $\mathcal{A}$ remains competitive with the optimal schedule by incurring a time delay.

\begin{lemma} \label{lem:alphabeta}
For any $k \geq 0$, we have $W^A_{\alpha \tau_{k+1}} \leq \beta W^{\mathcal{OPT}}_{\tau_k}$\enspace.
\end{lemma}
\begin{proof}
Every job completed by $\mathcal{OPT}$ by $\tau_k$ must have been released before $\tau_k$. For each such job $j$, either our algorithm completed it before time $\tau_k$ or $j \in R(\tau_k)$. The set of jobs completed by $\mathcal{OPT}$ by time $\tau_k$ gives a feasible solution to the MUWP with deadline $D = \tau_{k+1} - \tau_k = \tau_k $ and its total unscheduled weight is $W^{\mathcal{OPT}}_{\tau_k}$. Therefore, the optimal total unscheduled weight value for the MUWP when considering all $j \in R(\tau_k)$ with deadline $D$ is at most $W^{\mathcal{OPT}}_{\tau_k}$. By the definition of $(\alpha, \beta)$-approximation, the claim follows.
\end{proof} 

The next lemma states that ordering jobs within each interval further approximates the optimal schedule closely. For a fixed subset $S$ of jobs, let $\mathcal{OPT}(S)$ denote the optimal schedule for $S$ and $C^{\mathcal{OPT}(S)}_j$ denote the completion time of job $j$ in $\mathcal{OPT}(S)$. Also, let $\mathcal{OPT}_0(S)$ denote an optimal schedule that starts at time 0 and ignores all job release times, and let $C^{\mathcal{OPT}_0(S)}_j$ denote the completion time of job $j$ in $\mathcal{OPT}_0(S)$. 

\begin{lemma} \label{lem:local_delta}
The weighted completion time for schedule $\mathcal{OPT}_0(S)$ is at most that of schedule $\mathcal{OPT}(S)$; i.e., 
\begin{equation*}
\sum_{j \in S}w_{j} C_{j}^{\mathcal{OPT}_0(S)}\leq \sum_{j \in S}w_{j}C^{\mathcal{OPT}(S)}_j~~\enspace.
\end{equation*}
\end{lemma}
\begin{proof}
The optimal schedule of $S$ with release times defines a valid schedule for $S$ without release times, so the claim follows.
\end{proof}

\begin{figure}
\centering
\captionsetup{width=0.8\textwidth}
\includegraphics[scale=0.75]{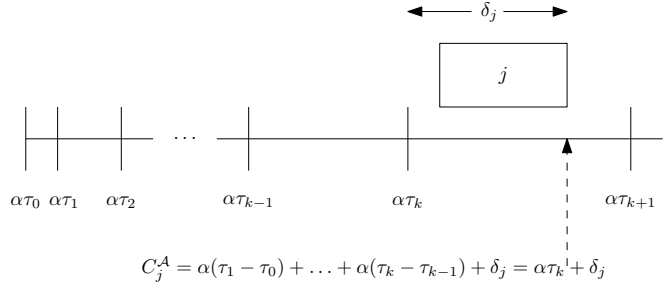}
\caption{We let $\delta_j$ denote the distance between $C^{\mathcal{A}}_j$ and the beginning of the interval in which job $j$ completes.}
\label{fig:delta} 
\end{figure}

Recall that at each $\tau_k$, $\mathcal{A}$ uses an $(\alpha, \beta)$-approximation on the MUWP to select a subset $S_k$ of $R(\tau_k)$ to schedule within $\alpha I_k$ using a $\gamma$-approximation that ignores release times. Let $C_j^\mathcal{A}$ denote the completion time of job $j$ in the schedule produced by $\mathcal{A}$, $t(j)$ denote the largest index such that job $j$ begins processing after time $\tau_{t(j)}$, and $\delta_j = C_j^{\mathcal{A}} - \alpha \tau_{t(j)}$ for each job $j \in J$ (see Fig.~\ref{fig:delta}). Let $L$ be the smallest time index such that the optimal schedule finishes by time $\tau_L$, and let $S_k$ denote the set of jobs scheduled independently by $\mathcal{A}$ in the interval $\alpha I_k$. Then Lemma~\ref{lem:local_delta} implies
\begin{equation} \label{eq:wjdj_bound}
\sum_{j \in S_k}w_{j} \delta_{j} \leq \gamma \sum_{j \in S_k}w_{j}C^{\mathcal{OPT}_0(S_k)}_j \leq \gamma \sum_{j \in S_k}w_{j}C_{j}^{\mathcal{OPT}(S_k)}\enspace.
\end{equation}

\begin{lemma} \label{lem:global_delta}
The weighted sum of the $\delta_j$ is at most $\gamma$ times the optimal weighted completion time; i.e.,
\[
\sum_{j \in J}w_j\delta_j \leq \gamma\sum_{j \in J}w_j C_j^{\mathcal{OPT}}\enspace.
\]
\end{lemma}

\begin{proof}
Recall that $S_1, \ldots, S_L$ partition $J$, and notice that due to \eqref{eq:wjdj_bound}, we have
\begin{equation} \label{eq:dj_to_opt}
\sum_{k=1}^L\sum_{j\in S_k}w_j\delta_j \leq \gamma\sum_{k=1}^L\sum_{j \in S_k} w_j C_j^{\mathcal{OPT}(S_k)} \leq \gamma\sum_{j \in J} w_j C_j^{\mathcal{OPT}}\enspace,
\end{equation}
thus proving the lemma.
\end{proof}

\begin{proof}[Proof of Theorem~\ref{thm:alphabetagamma}]
We rewrite the total weighted completion time of the schedule produced by $\mathcal{A}$ to obtain the following.
\begin{align*}
\sum_{j \in J}w_j C_j^\mathcal{A} &= \alpha \sum_{k=1}^L(\tau_k - \tau_{k-1})W_{\alpha \tau_k}^\mathcal{A} + \sum_{j \in J}w_j\delta_j \\
&= \alpha \sum_{k=2}^L(\tau_k - \tau_{k-1})W_{\alpha \tau_k}^\mathcal{A} + \alpha W_{\alpha \tau_1}^\mathcal{A} + \sum_{j \in J}w_j \delta_j \\
&\leq 2\alpha \sum_{k=1}^{L}(\tau_{k} - \tau_{k-1})W_{\alpha \tau_{k+1}}^\mathcal{A} + \alpha W + \sum_{j\in J} w_j\delta_j \\
&\leq 2\alpha \beta \sum_{k=1}^{L}(\tau_{k} - \tau_{k-1})W_{ \tau_{k}}^\mathcal{OPT} + \alpha W + \sum_{j\in J}w_j\delta_j \\
&\leq (2\alpha\beta + \gamma)\sum_{j\in J}w_j C_j^\mathcal{OPT} + \alpha W\enspace,
\end{align*}
where the last two inequalities follow from Lemmas~\ref{lem:alphabeta} and~\ref{lem:global_delta}, respectively.
\end{proof}

\section{Applications to Concurrent Open Shop}
\label{sec:cos}
Now we apply our minimization framework to $\infdiv{PD}{}{\sum_j{w_jC_j}}$. In Section~\ref{sec:cos_exp_details}, we give an offline dynamic program that optimally solves $\infdiv{PD}{r_j=0}{\sum_j w_j C_j}$ in exponential time, giving $\gamma = 1$ in our framework.

For the MUWP, in exponential time, we can iterate over every subset of jobs to find a feasible schedule that minimizes the total weight of unscheduled jobs, so this is a $(1,1)$-approximation, giving $\alpha = \beta = 1$. Thus, Theorem~\ref{thm:alphabetagamma} yields the following, which improves upon the competitive ratio of 4 from Garg et al.~\cite{garg2007order}. 

\begin{corollary} \label{cor:cos_exp_3}
There exists an exponential time 3-competitive algorithm for the concurrent open shop setting.
\end{corollary}

In polynomial time, Garg et al.~\cite{garg2007order} provide a $(2,2)$-approximation for the MUWP, and Mastrolilli et al.~\cite{mastrolilli2010minimizing} provide a 2-approximation for offline version of $\infdiv{PD}{r_j=0}{\sum_j w_j C_j}$. These results with Theorem~\ref{thm:alphabetagamma} improve the ratio of 16 by Garg et al.~\cite{garg2007order}. We note that the additional additive term of $\alpha W$ in Theorem~\ref{thm:alphabetagamma} is smaller than the additive $3W$ term in the guarantees of Garg et al.~\cite{garg2007order}, for both the exponential-time and polynomial-time cases.

\begin{corollary}
\label{cor:cos_poly_10}
There exists a polynomial-time 10-competitive algorithm for the concurrent open shop setting.
\end{corollary}

When the number of machines $m$ is constant, there exists a polynomial-time $(1+\epsilon,1)$-approximation algorithm for the MUWP (see Section~\ref{sec:cos_4_poly_fixed_m}) by a reduction to the multidimesional knapsack problem. Furthermore, when $m$ is fixed, Cheng et al.~\cite{edwin2011polynomial} gave a PTAS for the offline $\infdiv{PD}{r_j=0}{\sum_j w_j C_j}$. Combining these results with Theorem~\ref{thm:alphabetagamma} yields the following.

\begin{corollary}
There exists a polynomial time, $(3+\epsilon)$-competitive algorithm for $\infdiv{PD}{}{\sum{w_jC_j}}$ when the number of machines is fixed.
\end{corollary}

\section{Applications to Coflow Scheduling}
\label{sec:coflow}
We now apply our framework to coflow scheduling, introduced by Chowdhury and Stoica~\cite{chowdhury2012coflow}. We are given a non-blocking network with $m$ input ports and $m$ output ports. A \emph{coflow} is a collection of parallel flows processed by the network. We represent a coflow $j$ by an $m \times m$ matrix $D^j = (d_{io}^j)_{i,o \in [m]}$, where $d_{io}^j$ denotes the integer amount of data to be transferred from input port $i$ to output port $o$ for coflow $j$. Each port can process at most one unit of data per time unit, and we assume that the transfer of data within the network is instantaneous.

The problem is to schedule a set of $n$ coflows, each with a non-negative weight $w_j$ and release time $r_j$, that minimizes the sum of weighted completion times, where the completion time of a coflow is the earliest time at which all of its flows have been processed. We denote this problem by $\infdiv{\mbox{Coflow}}{}{\sum_j{w_jC_j}}$.


As in Sect.~\ref{sec:cos}, in exponential time, we can iterate over all subsets of coflows to optimally solve the MUWP, giving a $(1,1)$-approximation. Moreover, Ahmadi et al.~\cite{ahmadi2016coflows} proposed a 4-approximation for offline coflow scheduling\footnote{Since permutation schedules are not necessarily optimal for coflow scheduling~\cite{chowdhury2014efficient}, even finding a \emph{factorial}-time optimal algorithm is nontrivial. For simplicity, we have chosen to use a polynomial-time algorithm to achieve Corollary~\ref{cor:exp-coflow}.}.

\begin{corollary}\label{cor:exp-coflow}
There exists an exponential-time $6$-competitive algorithm  for online coflow scheduling.
\end{corollary}

Furthermore, we can show that the polynomial-time $(2,2)$-approximation for the MUWP for $\infdiv{PD}{}{\sum_j{w_jC_j}}$ of Garg et al.~\cite{garg2007order} can be applied to coflow scheduling with the same approximation guarantees. Combined with the $4$-approximation of Ahmadi et al.~\cite{ahmadi2016coflows}, our framework yields the following.

\begin{corollary}
There exists a polynomial-time $12$-competitive algorithm for online coflow scheduling.
\end{corollary}

To show that the $(2,2)$-approximation for the MUWP for $\infdiv{PD}{}{\sum_j{w_jC_j}}$ of Garg et al.~\cite{garg2007order} can be applied to coflow scheduling with the same approximation guarantees, we recall the reduction from $\infdiv{\mbox{Coflow}}{}{\sum_j{w_jC_j}}$ to $\infdiv{PD}{}{\sum_j{w_jC_j}}$ given by Khuller and Purohit~\cite{khuller2016brief}. Given an instance of coflow scheduling $\mathcal{I}$, let $L_i^j = \sum_{o=1}^m d_{io}^j$ denote the total amount of data that co-flow $j$ needs to transmit through input port $i$, and similarly, we let $L_o^j = \sum_{i=1}^m d_{io}^j$ or output port $o$. From this, create a concurrent open shop instance $\mathcal{I}'$ with a set $M$ of $2m$ machines (one for each port) and a set $J$ of $n$ jobs (one for each coflow), with processing times $p_{sj}$ set equal to $L^j_s$ for job $j$ on machine $s$.

Now, the MUWP on $\mathcal{I}'$ can be formulated by the following integer program of Garg et al.~\cite{garg2007order}.

\begin{align*}
    \mbox{minimize }   & \sum_{j \in J} w_j (1 - x_j)\  \\
    \mbox{subject to } & \sum_{j \in J} p_{ij} x_j \leq D \quad \forall i \in M\\
                       & x_j \in \{0,1\} \quad \forall j \in J\enspace.
\end{align*}

Let $W'$ denote the optimal unscheduled weight for the MUWP on $\mathcal{I}'$, and define $W$ similarly. The algorithm of Garg et al.~\cite{garg2007order} solves the linear relaxation of this integer program to obtain an optimal fractional solution $\bar{x}$, and outputs the set of jobs $S' = \{j \in J \mid \bar{x}_j \geq \frac{1}{2} \}$. Letting $W^*$ denote the objective function value of an optimal solution of the LP relaxation, it is straightforward to check that the total processing time of $S'$ on any machine is at most $2D$, the total unscheduled weight is at most $2W^*$, and $W^* \leq W'$. Hence, the algorithm of Garg et al.~\cite{garg2007order} is indeed a $(2,2)$-approximation for the MUWP in the concurrent open shop environment. 

\begin{lemma} \label{lem:opt_coflow}
The optimal unscheduled weight for the MUWP on $\mathcal{I}'$ is at most that for the MUWP on $\mathcal{I}$; i.e.,
$W' \leq W\enspace.$
\end{lemma}
\begin{proof}
The proof is essentially identical to that of Lemma 1 in \cite{khuller2016brief}. Let $S$ be the optimal solution to the MUWP for $\mathcal{I}$. Then there exists a schedule of the coflows in $S$ such that every coflow completes by the deadline $D$. Now consider the set $S'$ of corresponding jobs in $\mathcal{I}'$. Processing job $j \in S'$ on machine $s$ whenever data is being processed for coflow $j \in S$ on port $s$ in the schedule for $S$ gives a schedule for $S'$ in which every job also completes by deadline $D$. Thus $S'$ is a feasible solution to the MUWP for $\mathcal{I}'$ with objective function value equal to that of the optimal solution $S$ to the MUWP for $\mathcal{I}$, and the claim follows.
\end{proof}

Let $S$ be the set of coflows in $\mathcal{I}$ corresponding to the jobs $S'$ defined above.

\begin{lemma}
In polynomial time, we can find a schedule for $S$ that completes by time $2D$, and whose total unscheduled weight is at most $2W$.
\end{lemma}
\begin{proof} We know that for any machine $s$ in $\mathcal{I}'$, $\sum_{j \in S'} p_{sj} = \sum_{j \in S} L_s^j \leq 2D$. Thus, if we take any schedule for the coflows $S$ without idle time, every port $s$ finishes processing data by time $\sum_{j \in S} L_s^j \leq 2D$. Since all coflows complete at the same time when all ports have finished processing the data, we get a schedule in which all coflows in $S$ will complete without idle time by time $2D$.

The total unscheduled weight in $\mathcal{I}$ is the same as the total unscheduled weight in $\mathcal{I}'$. By Lemma~\ref{lem:opt_coflow}, the total unscheduled weight in $\mathcal{I}$ is at most $2W' \leq 2W$.
\end{proof}

Hence, the $(2,2)$-approximation for the MUWP for $\infdiv{PD}{}{\sum_j{w_jC_j}}$ of Garg et al.~\cite{garg2007order} can be applied to $\infdiv{\mbox{Coflow}}{}{\sum_j{w_jC_j}}$ with the same guarantees.

\section{A Randomized Online Scheduling Framework}
\label{sec:rand_framework}
In this section, we show how our ideas can be combined with the randomized framework of Chakrabarti et al.~\cite{chakrabarti1996improved} to develop an analogue of the deterministic framework of Sect.~\ref{sec:framework}.

The framework of Chakrabarti et al.~\cite{chakrabarti1996improved} modify that of Hall et al.~\cite{hall1997scheduling} (see Sect.~\ref{sec:hall_framework}) by setting $\tau_k = \eta 2^k$, where $\eta \in [\frac{1}{2},1)$ is a randomly chosen parameter. After making this choice, the online algorithm proceeds exactly as before, by applying the dual $\rho$-approximation to the MSWP at each interval. 

Let $C^{\mathcal{OPT}}_j$ denote the completion time of job $j$ in an optimal schedule, and let $B_j$ denote the start of the interval $(\tau_{k-1},\tau_k]$ in which job $j$ completes. Chakrabarti et al.~\cite{chakrabarti1996improved} show that if one takes $\eta = 2^{-X}$, where $X$ is chosen uniformly at random from $(0,1]$, then the following holds.

\begin{lemma}[\cite{chakrabarti1996improved}] \label{lem:rand_bound}
$E[B_j] = \frac{1}{2 \ln 2} C^{\mathcal{OPT}}_j\enspace.$
\end{lemma}

Hall et al.~\cite{hall1997scheduling} showed how to produce a schedule of total weighted completion time at most $4 \rho \sum_{j} w_j B_j$. By linearity of expectation, one can apply Lemma~\ref{lem:rand_bound}, so that the schedule produced has total weighted completion at most $\frac{2}{\ln 2} \rho \sum_{j} w_j C^{\mathcal{OPT}}_j$, resulting in a randomized $\frac{2}{\ln 2} \rho \leq 2.89 \rho$-competitive algorithm.

We can directly adapt this idea of randomly choosing the interval end points in our minimization framework to develop a randomized version of Theorem~\ref{thm:alphabetagamma}. Specifically, we take $\tau_k = \eta 2^k$, using the same $\eta$ above, and run the framework described in Sect.~\ref{sec:min_framework} using this new choice of interval end points. Note that Lemma~\ref{lem:alphabeta}, Lemma~\ref{lem:local_delta}, and Lemma~\ref{lem:global_delta} still hold with our new choice of $\tau_k$.

Let $\mathcal{A}'$ denote this randomized algorithm and using the same notation as in Sect.~\ref{sec:min_framework}, we can achieve the following result.

\begin{theorem} \label{thm:rand_alphabetagamma}
Algorithm $\mathcal{A}'$ is $(\frac{1}{\ln2}\alpha \beta + \gamma)$-competitive in expectation, with an additive $\alpha W$ term.
\end{theorem}
\begin{proof}
The same steps as in the proof of Theorem~\ref{thm:alphabetagamma} yield
\[
\sum_{j \in J} w_j C^{\mathcal{A}'}_j \leq 2\alpha \beta \sum_{k=1}^L(\tau_{k} - \tau_{k-1})W_{ \tau_{k}}^\mathcal{OPT} + \alpha W + \sum_{j\in J}w_j\delta_j\enspace.
\]
By definition of $B_j$, we notice that 
\[
\sum_{k=1}^L(\tau_{k} - \tau_{k-1})W_{ \tau_{k}}^\mathcal{OPT} = \sum_{j \in J} w_j B_j\enspace.
\]
By linearity of expectation and Lemma~\ref{lem:rand_bound}, we conclude that
\[
E[\sum_{j \in J} w_j C^{\mathcal{A}'}_j] \leq (\frac{1}{\ln 2}\alpha\beta + \gamma)\sum_{j\in J}w_j C_j^\mathcal{OPT} + \alpha W\enspace.
\]
\end{proof}

Using the guarantee of Theorem~\ref{thm:rand_alphabetagamma}, we can instantiate this framework in various scheduling settings and find values for $\alpha,\beta,\gamma$ to achieve improved competitive ratios over our deterministic framework. The results obtained when applying the same subroutines as we did for the framework of Sect.~\ref{sec:min_framework} are summarized in Table~\ref{table:results}.

\section{Applications to Concurrent Cluster Scheduling}
\label{sec:ccs}
We now describe the concurrent cluster model introduced by Murray et al.~\cite{murray2016scheduling} and give the first online algorithm in this setting. In the concurrent cluster model, there is a set $M$ of $m$ clusters, where each cluster $i \in M$ consists of a set of $m_i$ parallel machines. We are given a set of jobs $J$, where each job $j$ has $m$ subjobs, each corresponding to one of the clusters. Each subjob is specified by a set $T_{ji}$ of tasks to be completed by the machines in cluster $i$. Each task $t$ in $T_{ji}$ has a processing requirement of $p_{jit}$ time units. The subjobs of a given job may be processed concurrently across multiple clusters, and the tasks of a given subjob may be processed concurrently across multiple machines within the same cluster. A job is complete when all of its subjobs are complete, and a subjob is complete when all of its tasks are complete. A task can be completed by any machine within the cluster of that task. Furthermore, the speed of machines within a cluster may vary, but for our purposes, we assume that all of the machines in each cluster are identical. In the three field notation, we denote this by $\infdiv{CC}{v=1}{\sum_j w_j C_j}$, where $v = 1$ means the machines all have the same speed. The goal is to schedule the jobs to minimize total weighted completion time. Among other results for variants of this problem, Murray et al.~\cite{murray2016scheduling} gave an offline 3-approximation algorithm when all subjobs have identical release times.

Here, we present a $(4,2)$-approximation for the MUWP in the concurrent cluster setting. Using this in conjunction with the offline 3-approximation algorithm of Murray et al.~\cite{murray2016scheduling}, our framework gives the following online guarantee.

\begin{corollary} \label{ref:cc_19}
There exists a polynomial time 19-competitive algorithm for scheduling in the concurrent cluster model.
\end{corollary}

Now recall that the MUWP is the following: given a set of jobs and a deadline $D$, find a subset of the jobs that can be completed by time $D$ and minimizes the total weight of unscheduled jobs. The concurrent cluster setting differs from the settings we have already seen because it is nontrivial to even decide if a given subset of jobs can be scheduled to complete by a certain deadline.

For each job $j$ and subjob $i$ of $j$, let $P_{ji} = \sum_{t \in T_{ji}} p_{jit}$ denote the total processing requirements of the $i$-th subjob of job $j$, and let $B_{ji} = \max_{t \in T_{ji}} p_{jit}$ denote the processing time of the largest task of subjob $i$ of job $j$. We can then formulate the following LP relaxation of MUWP, which is similar to the LP relaxation of Garg et al.~\cite{garg2007order} in the concurrent open shop setting:

\begin{align*}
	\mbox{min} & \sum_{j\in J} w_j (1-x_j) \\
    \mbox{subject to } & \sum_{j \in J} \frac{P_{ji}}{m_i} x_j \leq D \quad \forall i \in M \\
                       & B_{ji} x_j \leq D \quad \forall j \in J, i \in M \\
		       & 0 \leq x_j \leq 1 \quad \forall j \in J\enspace.
\end{align*}

It is straightforward to see why this is a relaxation of the problem. Indeed, the first set of constraints requires that for each cluster $i$, the total processing time of all jobs in a feasible solution does not add up to more than the deadline. The second set of constraints requires that the largest task in each subjob can be completed on one machine by the deadline.

\begin{lemma} \label{lem:cc_22_lp_approx}
There exists a polynomial-time $(2,2)$-approximation to the above LP for the MUWP in the concurrent cluster model.
\end{lemma}
\begin{proof}
The proof is similar to that of Garg et al.~\cite{garg2007order}. We can solve the LP in polynomial time to obtain an optimal solution $\bar{x}$, and return the set $S = \{j \in J \mid \bar{x}_j \geq \frac{1}{2}\}$. It follows that the total processing time of $S$ is at most $2D$ and the total weight of jobs not in $S$ is at most twice the value of the optimum solution $\bar{x}$ of the LP.
\end{proof}

Note that the set obtained from Lemma~\ref{lem:cc_22_lp_approx} does not necessarily have a schedule which completes by time $2D$ in the MUWP for the concurrent cluster model. However, if we relax the deadline further from $2D$ to $4D$, then we can guarantee the existence of such a schedule.

\begin{lemma} \label{lem:cc_42}
There exists a feasible schedule of the set $S$ from Lemma~\ref{lem:cc_22_lp_approx} that completes by time $4D$.
\end{lemma}
\begin{proof}
On each cluster $i$, we apply the classical list scheduling algorithm for minimizing the makespan on a set of identical parallel machines. Thus, the makespan of cluster $i$ is at most $\sum_{j \in S} \frac{P_{ji}}{m_i} + \max_{j \in S} B_{ji}$. Furthermore, Lemma~\ref{lem:cc_22_lp_approx} implies that each of these terms is at most $2D$. So, the makespan of cluster $i$ is at most $4D$, and the lemma follows.
\end{proof}

Lemma~\ref{lem:cc_42} directly implies Corollary~\ref{ref:cc_19}.

\section{An optimal exponential-time offline algorithm for concurrent open shop}
\label{sec:cos_exp_details}
In this section, we present an offline exponential-time algorithm that optimally solves $\infdiv{PD}{r_j=0}{\sum{w_jC_j}}$. The algorithm is based on the Held-Karp algorithm (see~\cite{held1962dynamic}) for the traveling salesperson problem.

First, we state a definition and a lemma that restricts the number of possible schedules.

\begin{definition}[\cite{mastrolilli2010minimizing}]
A schedule is a \emph{permutation schedule} if it is nonpreemptive, contains no unnecessary idle time, and there exists a permutation of the jobs $\sigma: J \rightarrow J$ such that on every machine, the jobs are scheduled in the order specified by $\sigma$.
\end{definition}

Next, we rephrase a lemma from~\cite{mastrolilli2010minimizing} that allows us to restrict our attention to permutation schedules.

\begin{lemma}[\cite{mastrolilli2010minimizing}] \label{lem:cos_perm_sched}
For any schedule for concurrent open shop with identical release times, there exists a corresponding permutation schedule  in which the completion time of every job is no later than it is in the original schedule.
\end{lemma}

For any subset $S \subseteq J$ of jobs and job $j \in S$, let $C(S, j)$ denote the weighted completion time of the optimal schedule for $S$ that completes job $j$ last. If $S$ only contains one job $j$, then certainly there is only one feasible schedule, so we have that $C(S,j) = w_j \cdot \max_{i \in M} p_{ij}$. If $S$ contains at least two jobs, consider an optimal schedule for $S$ in which $j$ finishes last. Suppose job $i \in S \setminus \{j\}$ finishes second-to-last in this schedule. Since any contiguous subsequence of jobs in an optimal schedule is itself an optimal schedule for that subset of jobs,
\[
C(S,j) = C(S \setminus \{j\}, i) + w_j \max_{k \in M} \sum_{q \in S} p_{kq}.
\]

Algorithm~\ref{alg:dpforcos} contains the pseudocode for this algorithm. Using standard backtracking techniques, one can obtain an optimal schedule using this dynamic program.

\begin{algorithm}[H]
\caption{An optimal offline dynamic program for $\infdiv{PD}{r_j=0}{\sum{w_jC_j}}$}
\label{alg:dpforcos}
\begin{algorithmic}
\FOR{$j \gets 1, \ldots, n$}
	\STATE $C(\{j\} ,j) \gets w_j \max_{i \in M}p_{ij}$
\ENDFOR
\FOR{$t \gets 2, \ldots, n$}
	\FORALL{$S \subseteq J$ s.t. $|S| = t$}
    	\FORALL{$j \in S$}
            \STATE $prev \gets \min_{i \in S \setminus\{ j\}}\{C(S \setminus \{ j\}, i)\}$
            \STATE $curr \gets w_j \max_{k \in M} \sum_{q \in S} p_{kq}$
            \STATE $C(S,j) \gets prev + curr$
        \ENDFOR
    \ENDFOR
\ENDFOR

\STATE{\RETURN $\displaystyle \min_{j\in J}C(J, j)$}
\end{algorithmic}
\end{algorithm}

\begin{theorem} \label{thm:cos_exp_dp}
Algorithm~\ref{alg:dpforcos}  optimally solves the offline concurrent open shop with identical release times problem in $O(m n^2 2^n)$ time.
\end{theorem}

\begin{proof}
The optimality of the algorithm follows from the above discussion, and we briefly sketch the running time bound.

There are $2^n$ subsets $S \subseteq J$ and $O(n)$ jobs $j \in S$, so there are $O(n \cdot 2^n)$ subproblems $C(S,j)$. For each subproblem $C(S,j)$, the search over the possible penultimate jobs in Algorithm~\ref{alg:dpforcos} takes $O(mn)$ time, giving a final running time of $O(m n^2 2^n)$.
\end{proof}

\section{A polynomial-time $(1+\epsilon,1)$-approximation for MUWP in concurrent open shop when $m$ is fixed}
\label{sec:cos_4_poly_fixed_m}
In this section, we present a dual $(1+\epsilon)$-approximation algorithm for the MSWP for concurrent open shop. Recall that this is algorithm is also a $(1+\epsilon,1)$-approximation algorithm for the MUWP.

We can view the MSWP for concurrent open shop as a multidimensional version of the knapsack problem see ~\cite{frieze1984approximation}. We briefly reformulate the maximum scheduled weight problem as a generalization of the knapsack problem. We are given a set $J$ of $n$ items (jobs) in $m$ dimensions (machines). Each job $j$ has a non-negative weight $w_j$, and a non-negative size $p_{ij}$ in the $i^\text{th}$ dimension, where $i \in M$. We have a single knapsack (set of machines) which has non-negative capacity $s_i$ in the $i^\text{th}$ dimension. For the maximum scheduled weight problem, the capacity for each dimension is the same, that is, $s_1 = s_2 = ... = s_m = D$, where $D$ denotes the common deadline on every machine. We want to find a maximum-weight subset of items that can be packed into the knapsack without exceeding the knapsack capacity in any dimension.

This problem can be formulated as the following integer program $P$:
\begin{align*}
	\mbox{maximize} & \sum_{j\in J} w_jx_j \\
    \mbox{subject to } & \sum_{j \in J} p_{ij} x_j \leq s_i \quad \forall i \in M \\
                       & x_j \in \{0,1\}, \forall j \in J\enspace.
\end{align*}

Our dual $(1+\epsilon)$-approximation algorithm for $P$ runs a dynamic program on a scaled version $\widehat{P}$ of $P$, and returns the resulting set of jobs (see Hall et al.~\cite{hall1997scheduling} for a similar approach when $m=1$). For every $i \in M, j \in J$, let $\widehat{p_{ij}} = \floor{\frac{p_{ij}}{b_i}}$ and $\widehat{s_i} = \ceil{\frac{s_i}{b_i}}$, where $b_i = \frac{\epsilon}{n+1}s_i$, and substitute these values in $P$ to obtain the new integer program $\widehat{P}$. The weights of the jobs remain unchanged.

Next, we state a well-known dynamic program for the multidimensional knapsack problem and appy it to $\widehat{P}$. Let $OPT(j, \widehat{s_1}, \widehat{s_2}, \ldots, \widehat{s_m})$ denote the optimum value of $\widehat{P}$ when restricted to the items $\{1, \ldots, j\}$ and the knapsack has remaining capacity $\widehat{s_i}$ in dimension $i \in M$. Then $OPT(j, \widehat{s_1}, \widehat{s_2}, \ldots, \widehat{s_m})$ is equal to
\[
OPT(j-1, \widehat{s_1}, \widehat{s_2}, \ldots,\widehat{s_m})
\]
if $s_i < p_{ij}$ for some $i \in M$ and
\begin{align*}
\max\{&OPT(j-1, \widehat{s_1}, \widehat{s_2}, \ldots, \widehat{s_m}), \\ 
&w_j + OPT(j-1, \widehat{s_1}-\widehat{p_{1j}}, \widehat{s_2}-\widehat{p_{2j}}, \ldots, \widehat{s_m}-\widehat{p_{mj}})\}
\end{align*}
otherwise. 

\begin{lemma}\label{lem:dpruntime}
The above dynamic program runs in $O(m\left(\frac{n}{\epsilon}\right)^m)$ time.
\end{lemma}

\begin{proof}
Computing the optimal solution value $OPT(n, \widehat{s_1}, \widehat{s_2}, \ldots, \widehat{s_m})$ can be seen as filling a table of size $n\cdot \widehat{s_1} \cdot \ldots \cdot \widehat{s_m}$. Since each entry takes $O(m)$ time to compute, the lemma follows.
\end{proof}

Let $S$ denote the subset of jobs obtained from solving $\widehat{P}$.

\begin{lemma} \label{lem:dplem}
The optimum value of $\widehat{P}$ is at least that of $P$. Furthermore, for every $i \in M$, the total size of the items in the $i$-th dimension of $S$ is at most $(1+\epsilon) s_i$.
\end{lemma}

\begin{proof}
To obtain $\widehat{P}$ from $P$, we rounded the $p_{ij}$ down and $s_i$ up. Thus, every feasible solution of $P$ is feasible for $\widehat{P}$, so the first statement follows.

For every $i \in M$ and $j \in J$, let $\overline{s_i} = \widehat{s_i} \cdot b_i = \ceil{\frac{s_i}{b_i}}b_i$, and let $\overline{p_{ij}} = \widehat{p_{ij}} \cdot b_i = \floor{\frac{p_{ij}}{b_i}}b_i$, so that $s_i+b_i \geq \overline{s_i}$ and $p_{ij} \leq \overline{p_{ij}}+b_i$. Then for each dimension $i \in M$,
\begin{align*}
\sum_{j \in S} p_{ij} &= \sum_{j \in T} (p_{ij} + \overline{p_{ij}} - \overline{p_{ij}}) \\
&= \sum_{j \in S} \overline{p_{ij}} + \sum_{j \in S}(p_{ij} - \overline{p_{ij}}) \\
&\leq \sum_{j \in S} \overline{p_{ij}} + n \cdot b_i \\
&\leq \overline{s_i} + n \cdot b_i \\
&\leq s_i + b_i + n \cdot b_i \\
&= s_i+(n+1) \cdot \frac{\epsilon}{n+1}s_i \\
&= (1+\epsilon) s_i\enspace.
\end{align*}
\end{proof}

\begin{theorem} \label{thm:cos_fixedm}
The algorithm described above runs in $O((\frac{n}{\epsilon})^m m)$ time and is a dual $(1+\epsilon)$-approximation algorithm for the maximum scheduled weight problem in the concurrent open shop setting.  
\end{theorem}

\begin{proof}
The proof follows directly from Lemma~\ref{lem:dpruntime} and Lemma~\ref{lem:dplem}. Note that this is a polynomial-time algorithm when $m$ is fixed.
\end{proof}

\medskip
{\bf Acknowledgements.}
We would like to thank Sungjin Im and Clifford Stein for directing us to \cite{chakrabarti1996improved,lubbecke2016new}, and William Gasarch for organizing the REU program.

\bibliographystyle{plainurl}
\bibliography{abbrevbib.bib}

\end{document}